\documentclass[11pt,reqno]{amsart}
\pdfoutput=1
\usepackage[margin=2.3cm]{geometry}
\usepackage[ascii]{inputenc}
\usepackage[dvipsnames]{xcolor}
\usepackage[bookmarksnumbered,linktocpage,hypertexnames=false,colorlinks=true,linkcolor=NavyBlue,urlcolor=NavyBlue,citecolor=ForestGreen,anchorcolor=green,breaklinks=true,pagebackref=true,pdfusetitle]{hyperref}
\usepackage{bbm,braket,microtype,mathrsfs,amsmath,amssymb,amsthm,amsfonts,latexsym,mathtools,graphicx,enumitem,booktabs,bm,xspace,float,mathdots,caption,subcaption,ellipsis,mleftright}
\usepackage[T1]{fontenc}
\usepackage{textcomp}
\usepackage[english]{babel}
\usepackage[capitalize,nameinlink]{cleveref}
\newtheorem{theorem}{Theorem}[section]
\newtheorem*{theorem*}{Theorem}
\newtheorem{lemma}[theorem]{Lemma}
\newtheorem{proposition}[theorem]{Proposition}
\newtheorem{corollary}[theorem]{Corollary}

\theoremstyle{definition}
\newtheorem{definition}[theorem]{Definition}
\theoremstyle{remark}
\newtheorem{remark}[theorem]{Remark}

\AddToHook{env/lemma/begin}{\crefalias{theorem}{lemma}}
\AddToHook{env/proposition/begin}{\crefalias{theorem}{proposition}}
\AddToHook{env/corollary/begin}{\crefalias{theorem}{corollary}}
\AddToHook{env/conjecture/begin}{\crefalias{theorem}{conjecture}}
\AddToHook{env/definition/begin}{\crefalias{theorem}{definition}}
\AddToHook{env/remark/begin}{\crefalias{theorem}{remark}}
\AddToHook{env/example/begin}{\crefalias{theorem}{example}}
\crefname{appendix}{Appendix}{Appendices}
\numberwithin{equation}{section}
\DeclareMathOperator{\tr}{tr}

\DeclareMathOperator{\supp}{supp}
\DeclareMathOperator{\id}{id}
\DeclareMathOperator{\cb}{cb}

\DeclareMathOperator{\image}{image}
\DeclareMathOperator{\round}{round}

\newcommand{\ot}{\otimes}

\newcommand{\eps}{\varepsilon}

\newcommand{\cA}{\mathcal A}
\newcommand{\cB}{\mathcal B}
\newcommand{\cC}{\mathcal C}

\newcommand{\cH}{\mathcal H}

\newcommand{\cL}{\mathcal L}
\newcommand{\cR}{\mathcal R}
\newcommand{\cX}{\mathcal X}
\newcommand{\cY}{\mathcal Y}

\newcommand{\cM}{\mathcal M}
\newcommand{\bigO}{\mathcal O}

\DeclarePairedDelimiter\abs{\lvert}{\rvert}
\DeclarePairedDelimiter\norm{\lVert}{\rVert}

\allowdisplaybreaks[4]
\newcommand{\approxSub}[1]{\subset_{#1}}

\newcommand{\approxEq}[1]{=_{#1}}
\newcommand{\subsetsim}{\mathrel{\substack{\subset\\[-0.35ex]\mathclap{\sim}}}}
\newcommand{\distloc}{\mathrm{dist}_\mathrm{loc}}
\newcommand{\distlocH}{\mathrm{dist}_\mathrm{loc}^H}
\newcommand{\approxIn}[1]{\in_{#1}}

\begin{document}

\title{Approximate QCAs in one dimension using approximate~algebras}
\date{}
\author{Daniel Ranard}
\address{California Institute of Technology, USA}
\email{dranard@caltech.edu}
\author{Michael Walter}
\address{Ludwig-Maximiliations-Universit\"at M\"unchen \& Munich Center for Quantum Science and Technology (MCQST), Germany}
\email{michael.walter@lmu.de}
\author{Freek Witteveen}
\address{CWI, Amsterdam, The Netherlands}
\email{F.Witteveen@cwi.nl}
\begin{abstract}
Quantum cellular automata (QCAs) are automorphisms of tensor product algebras that preserve locality, with local quantum circuits as a simple example.
We study approximate QCAs, where the locality condition is only satisfied up to a small error, as occurs for local quantum dynamics on the lattice.
A priori, approximate QCAs could exhibit genuinely new behavior, failing to be well-approximated by any exact QCA.
We show this does not occur in one dimension: every approximate QCA on a finite circle can be rounded to a strict QCA with approximately the same action on local operators, so these systems are classified by the same index as in the exact case.
Previous work considered the case of the infinite line, by using global methods not amenable to finite systems.
Our new approach proceeds locally and now applies to finite systems, including circles or homomorphisms from sub-intervals.  We extract exact local boundary algebras from the approximate QCA restricted to local patches, then glue these to form a strict~QCA.
The key technical ingredient is a robust notion of the intersection of two subalgebras: when the projections onto two subalgebras approximately commute, we construct an exact subalgebra that serves as a stable proxy for their intersection.
This construction uses a recent theorem of Kitaev on the rigidity of approximate $C^*$-algebras.
\end{abstract}
\maketitle

\section{Introduction}

Quantum cellular automata (QCAs) on tensor product algebras are automorphisms whose action on operators preserves locality.
They are also called locality-preserving unitaries.
In one dimension, they are characterized by the GNVW index~\cite{gross2012index}, and all 1D QCAs are compositions of circuits and shifts (translations).
In two dimensions, QCAs are likewise composed of circuits and generalized shifts, while their classification in three dimensions and higher is much richer \cite{freedman2020classification,shirley2022three,haah2023nontrivial,haah2023invertible,fidkowski2025quantum,haah2025topological,jones2024quantum,sun2025clifford} 

We study \emph{approximate} quantum cellular automata, whose action on operators preserve locality up to some small error, motivated by the approximate Lieb--Robinson ``light cones'' omnipresent in nonrelativistic local quantum dynamics.
A priori, the classification of approximate QCAs could be different than that of strict QCAs.
In one direction, one could imagine that any two QCAs are connected by a path of approximate QCAs, rendering the classification trivial.
In another direction, there could be approximate QCAs that are not well-approximated by any strict QCA, even when compared only locally.
Such a discrepancy in the classification of strict versus approximate objects is conjectured to occur in the closely related subject of invertible topological phases.
In particular, it is believed that there exist invertible states with exponentially decaying correlations (``tails'') such that every state in the same phase must have also have tails: no invertible states in the same phase have strictly zero correlations beyond a finite scale.
One such candidate example is Kitaev's $E_8$ state in two dimensions.
Whether similar phenomena can occur in the study of QCAs is an open question, which we resolve here in the negative for the case of 1D QCAs.

The study of approximate QCAs was initiated in Ref.~\cite{ranard2022converse} by the present authors, for the case of QCAs defined on the infinite line.
For this special case, the classification of approximate QCAs was shown to match that of strict QCAs, characterized by a robust version of the GNVW index.
The proof relied on partitioning the infinite line into two infinite half-lines, working with the associated infinite-dimensional algebras.
However, these global techniques did not allow any statements about approximate QCAs for systems defined on a \textit{finite} number of sites, e.g., a finite circle.

In this work, we develop stronger, \emph{local} techniques for studying approximate QCAs in one dimension, including on the circle.
Crucially, we are able to extract exact ``boundary algebras'' \cite{freedman2020classification, haah2023invertible} from the approximate QCA restricted to a local patch.
Our techniques also apply to QCA-like homomorphisms that are only defined locally, without assuming they arise from the restriction of some global automorphism.
We hope that the study of approximate QCAs in higher dimension may benefit from this more fine-grained control in one dimension.
 For instance, a QCA on a torus also yields a QCA on a circle, by coarse-graining the transverse dimensions.
 More generally, exact QCAs in all dimensions are characterized by their boundary algebras~\cite{haah2023invertible}.

We find that one-dimensional approximate QCAs in finite volume are essentially classified in the same way as strict QCAs, by the GNVW index.
We emphasize that perhaps counterintuitively, the case of finite volume is more difficult than the case of infinite volume.
While the latter requires more analytical care, certain difficulties are pushed to infinity.
More concretely, the infinite half-line considered by Ref.~\cite{ranard2022converse} has a $0$-dimensional boundary with a \textit{single} connected component, which aids in the extraction of the boundary algebra. In contrast, finite sub-intervals have two boundary components, rendering the extraction of the boundary algebra more difficult.

Our analysis in finite volume makes use of a difficult recent theorem of Kitaev~\cite{kitaev2024almost}, concerning a notion of approximate $C^*$-algebras.
(The present work was a partial motivation for Ref.~\cite{kitaev2024almost}, in addition to related applications in the study of topological phases.)
Using this result, we develop a robust notion of the intersection of two subalgebras.  Then we can extract boundary algebras of approximate QCAs, after casting them as approximate intersections.

This local approach can also be used to reproduce the results for infinite systems in Ref.~\cite{ranard2022converse}, though that is not our focus. We work in an elementary setting where the ``tails'' of the approximate QCA are not assumed to decay with distance; instead, we only assume they are small beyond some fixed radius.

\subsection{Organization of the paper}
In \cref{sec:qcas}, we give definitions of exact and approximate QCAs and discuss how to measure their closeness.
In \cref{sec:summary of results}, we state the main results of this work (\cref{thm:circle-round,thm:round-homomorphism,thm:index,thm:discrete path}).
In \cref{sec:notation}, we introduce some basic notions of approximate equality and inclusions.
In \cref{sec:locally-rounding}, we prove a rounding theorem for local QCAs (\cref{thm:local-rounding}) and use it to deduce our main results.
In \cref{sec:approx algebras}, we prove a theorem exhibiting a robust notion of the intersection of two subalgebras (\cref{lemma:intersections}), which serves as a key technical ingredient of the rounding theorem.
Finally, \cref{app:near-inclusion-lemmas} collects facts about near inclusions of algebras used in our proofs.

\subsection{Acknowledgments}
We are grateful to Alexei Kitaev for sharing a preprint of Ref.~\cite{kitaev2024almost}.
MW acknowledges support by the European Union (ERC Grant SYMOPTIC, 101040907) by the Deutsche Forschungsgemeinschaft (DFG, German Research Foundation, 556164098), by the Deutsche Forschungsgemeinschaft (DFG, German Research Foundation) under Germany's Excellence Strategy~--~EXC-2111~--~390814868, and by the German Federal Ministry of Research, Technology and Space (QuSol, 13N17173).
DR acknowledges support by the Simons Foundation under grant 376205.
Part of this work was conducted while MW was visiting Q-FARM and the Leinweber Institute for Theoretical Physics at Stanford University and the Simons Institute for the Theory of Computing at UC Berkeley.

\section{QCAs and approximate QCAs}\label{sec:qcas}

\subsection{Definitions} \label{sec:QCA-definitions}
We introduce quantum cellular automata (QCAs) with some generality for the purpose of illustration, before focusing on one spatial dimension.
Let $\Omega$ denote a finite set of ``sites'' with a metric, perhaps a subset of a manifold with metric.
Most of our notions are easily adapted to an infinite number of sites, because we only work with constant-size subsets.
Assign a finite-dimensional simple $C^*$-algebra~$\cA_x$ to each site~$x \in \Omega$.
That is, $\cA_x$ is isomorphic to the algebra of operators on $\mathbb{C}^{d_x}$ for some $d_x \in \mathbb{N}$.
Then the global algebra of observables is $\cA_\Omega = \bigotimes_{x \in \Omega} \cA_x$ (we will sometimes abbreviate $\cA=\cA_\Omega$).
For a subset of sites $X \subset \Omega$, we let
\begin{equation*}
    \cA_X = \bigotimes_{x \in X} \cA_x \subset \cA_\Omega.
\end{equation*}
For $X \subset \Omega$, let $X^{+r} = \{ y  \in \Omega : \exists x \in X \text{ s.t.\ } \mathrm{dist}(x,y) \le r \}$ denote the neighborhood of radius~$r$.
For~$x \in \Omega$, we abbreviate $x^{+r} = \{x\}^{+r}$.
When~$r$ is fixed to a constant, we also denote $X^{+} = X^{+r}$.

Now we can define a notion of locality-preserving homomorphism.
Here, a \emph{homomorphism} always refers to a unital $*$-homomorphism of $C^*$-algebras.
Note that when the domain is simple, such as $\cA_X$, such a homomorphism is always an injective isometry.

\begin{definition}\label{def:local-homomorphism}
    For $S \subset \Omega$, an \emph{(exactly) locality-preserving homomorphism} with range~$r$ is a homomorphism $\alpha : \cA_S \to \cA_\Omega$ such that, for each $X \subset S$,
    \begin{align}\label{eq:exact-LP}
        \alpha(\cA_X) \subset \cA_{X^{+r}}.
    \end{align}
    When $S=\Omega$, then $\alpha$ is automatically an automorphism, and we call it a \emph{quantum cellular automaton (QCA)} or a \emph{locality-preserving automorphism}.
\end{definition}
\noindent
To obtain \cref{eq:exact-LP}, note it is sufficient to demand locality on single sites: $\alpha(\cA_x) \subset \cA_{x^{+r}}$ for all $x \in S$.
If $\alpha$ is a QCA, then $\alpha^{-1}$ is a QCA of the same range, shown by applying \cref{eq:exact-LP} to complements, applying $\alpha^{-1}$, and taking commutants.

For subalgebras $\cX, \cY \subset \cA$, we write $\cX \subset_\eps \cY$ to denote an approximate inclusion,
\begin{equation*}
    \cX \approxSub{\eps} \cY  \quad \Leftrightarrow \quad  \forall x \in \cX, \exists y \in \cY \text{ s.t. } \norm{x-y} \leq \eps \norm{x}.
\end{equation*}
We discuss approximate inclusions and further notation along these lines in \cref{sec:notation}.
Now we can introduce a straightforward generalization of \cref{def:local-homomorphism} to the approximate setting.

\begin{definition}\label{def:eps-local-homomorphism}
    For $S \subset \Omega$, an \emph{($\eps,r$)-locality-preserving homomorphism} is a homomorphism $\alpha : \cA_S \to \cA_\Omega$ such that, for each $x \in S$,
    \begin{align} \label{eq:approx-LP}
        \alpha(\cA_x) \subset_\eps \cA_{x^{+r}}
    \end{align}
    We omit $r$ when it is fixed to some constant.
    When $S = \Omega$, then $\alpha$ is automatically an automorphism, and we call it an $\eps$-quantum cellular automaton (\emph{$\epsilon$-QCA}) or an \emph{$\epsilon$-locality-preserving automorphism}.
\end{definition}
\noindent
When $\epsilon=0$, we recover \cref{def:local-homomorphism}.

We emphasize that \cref{def:eps-local-homomorphism} does not constrain the ``far tails'' of the homomorphism: we do \emph{not} demand $\alpha(\cA_x)$ is increasingly well-approximated by algebras on increasing subsets, only that $\alpha(\cA_x)$ is localized up to a fixed error~$\eps$ to the neighborhood $x^{+r}$ for fixed~$r$. This weak assumption on the tail will suit our present purposes.

 Meanwhile, \cref{def:eps-local-homomorphism} \textit{does} automatically constrain the locality of $\alpha(\cA_S)$ for larger subsets $S$, albeit in a weak way.
 In particular:
\begin{proposition}\label{prop:extend-locality}
    An $(\epsilon,r)$-locality-preserving homomorphism $\alpha : \cA_S \to \cA_\Omega$ satisfies, for each $X \subset S$,
    \begin{equation*}
        \alpha(\cA_{X}) \subset_\delta \cA_{X^{+r}}, \quad \text{where} \quad \delta = 4 \epsilon |X|.
    \end{equation*}
\end{proposition}
\noindent
This follows from \cref{lemma:simultaneous-near-inclusions} applied to $\alpha(\cA_x) \subset_\eps \cA_{X^{+r}}$ for each $x \in X$.

We can define a distance on homomorphisms $\alpha,\beta : \cA_S \to \cA_\Omega$ by comparing them only on single-site operators:
\begin{equation}\label{eq:dist-local}
\distloc(\alpha,\beta)
:= \sup_{\substack{\norm{x}=1,\\ |\supp(x)|=1}}\;\|\alpha(x)-\beta(x)\|,
\end{equation}
where $|\supp(x)|=1$ denotes a single-site operator.
It may be useful to compare this to a metric discussed by Haah~\cite[App.~A]{haah2023invertible}, with $H$ indicating Haah,
\begin{equation*}
\distlocH(\alpha,\beta)
:= \sup_{\norm{x}=1}\;
\frac{\|\alpha(x)-\beta(x)\|}{|\mathrm{supp}(x)|}.
\end{equation*}
For homomorphisms, these metrics are comparable within a constant,
\begin{equation*}
\distloc(\alpha,\beta) \le \distlocH(\alpha,\beta)\le 2\sqrt{2}\;\distloc(\alpha,\beta).
\end{equation*}
The first inequality is immediate, and the second follows directly from \cref{lemma:local-to-global-hom}.

An $\eps$-locality-preserving homomorphism defined on $S$ need not be approximately surjective, even onto the interior of $S$.
For instance, it may leave a ``hole,'' sending all operators away from a given site.
We therefore define a notion of approximate local surjectivity:

\begin{definition}
    For $S \subset \Omega$, a homomorphism $\alpha : \cA_S \to \cA_\Omega$ is called \emph{$(\epsilon,r)$-locally-surjective} if, for each~$x \in S$ such that $x^{+r} \subset S$, it holds that
    \begin{equation*}
    \cA_x \subset_\eps \alpha(\cA_{x^{+r}}).
    \end{equation*}
    We omit $r$ when it is fixed to some constant.
    If $\eps=0$, this is called an \emph{(exactly) locally surjective} homomorphism.
\end{definition}
\noindent
This property also extends automatically to larger subsets $V$, at the cost of a volume factor $|V|$:
\begin{proposition} \label{prop:extend-loc-surjectivity}
An $(\epsilon,r)$-locally-surjective homomorphism $\alpha : \cA_S \to \cA_\Omega$ satisfies, for each $X \subset S$ such that~$X^{+r} \subset S$,
\begin{equation*}
    \cA_X \subset_\delta \alpha(\cA_{X^{+r}}), \quad \text{where} \quad \delta = 4 \epsilon |X|.
\end{equation*}
\end{proposition}
\noindent This follows from \cref{lemma:simultaneous-near-inclusions} applied to $\cA_x \subset_\eps \alpha(\cA_{V^{+r}})$ for each $x \in V$.

If $\alpha$ has a right inverse that is an $(\eps,r)$-locality-preserving homomorphism, then $\alpha$ is also $(\eps,r)$-locally-surjective.
Conversely, since $\alpha$ is injective and finite-dimensional, it admits a left inverse (which may or may not be a homomorphism); if $\alpha$ is $(\eps,r)$-locally-surjective, then any such left inverse is $(\eps,r)$-locality-preserving.

\subsection{Obstacles to naive rounding}
Given some notion of an approximate QCA $\alpha$, we are interested in whether it can be approximated by a strict QCA $\tilde \alpha$ with small $\distloc(\alpha,\tilde \alpha)$.
To appreciate the potential difficulty of this problem, even using this weak notion of distance, let us examine some possible strategies.
One possible strategy might proceed as follows: If
\[
\alpha(\cA_X) \subset_\eps \cA_{X^+},
\]
we can find some nearby $\tilde \alpha$ whose image on this region is strictly localized,
\[
\tilde \alpha(\cA_X) \subset \cA_{X^+}.
\]
To see this, note there exists $u$ such that $u \alpha(\cA_X) u^* \subset \cA_{X^+}$ and $\norm{u-I}=O(\eps)$, guaranteed to exist by \cref{thm:near-inclusion-unitary}; then we take $\tilde \alpha(x) = u \alpha(x) u^*.$
Moreover, one can take $u$ in the algebra generated by $\cA_{X^+}$ and $\alpha(\cA_X)$.
However, it is unclear how to continue.
If $\tilde \alpha$ were further modified by the same method to strictly localize the image of some overlapping region, say enforcing $\tilde \alpha(Y) \subset \cA_{Y^+}$ for $X \cap Y \neq \emptyset$,
then generally the previous condition $\tilde \alpha(\cA_X) \subset \cA_{X^+}$ would no longer hold exactly.

We sketch another method which fails, but not entirely: given an approximate QCA $\alpha$, we can approximate it with a unital CP map $F$ that is strictly local, i.e.\ $F(\cA_X) \subset \cA_{X^{+r}}$, but not a homomorphism.
To see this, first note that for the case of a strict QCA $\alpha$ of range $r$, the automorphism $\alpha \otimes \alpha^{-1}$ on the doubled system $\cA_1 \otimes \cA_2$ can be expressed as (conjugation by) a product of commuting unitaries, each with support of radius $r$ on the doubled system~\cite[Prop.~1]{gross2012index}.
For suitable $\Omega$, this can also be organized as a constant-depth circuit of unitaries with disjoint support.
These unitaries take the form
\begin{equation*}
    u_x = s_x (\alpha \otimes \id)(s_x)
\end{equation*}
where $s_x$ is the unitary that swaps the two copies of $\cA_x$. When $\alpha$ is $\eps$-locality-preserving, the unitaries~$u_x$ can be approximated by strictly local unitaries (\cref{prop:unitary-near-subalgebra}), composing a constant-depth circuit.
The circuit yields a strict QCA
\begin{equation*}
    \beta : \cA_1 \otimes \cA_2 \to \cA_1 \otimes \cA_2, \qquad \beta \approx \alpha \otimes \alpha^{-1}
\end{equation*}
(This approximation has errors independent of local Hilbert space dimension but depending on volume $|\Omega|$; the latter is removed if one uses stronger assumptions on the tails of $\alpha$.)
Let $\pi_1$ denote the Hilbert-Schmidt projection $\cA_1 \otimes \cA_2 \to \cA_1$, and let $\iota_1 : \cA_1 \to \cA_1 \otimes \cA_2$ denote the inclusion.
Then
\begin{equation*}
    F = \pi_1 \circ \beta \circ \iota_1 : \cA_1 \to \cA_1
\end{equation*}
is a unital CP map which approximates $\alpha$ but is strictly local.  While $F$ could be modified to an exact homomorphism (\cref{thm:make homomorphism exact}), it would no longer be strictly local.

In summary, it is easy to modify unitary circuits with approximately local gates to obtain strictly local gates.
 Hence it is easy to make approximate QCAs exactly local --- at the price of working with $\alpha \otimes \alpha^{-1}$ on the doubled system, or working with a unital CP map on the original system that is only an approximate homomorphism.

The strategy we ultimately use to approximate $\alpha$ by an exact QCA $\tilde \alpha$ is based on the classification of one-dimensional QCAs. This classification identifies ``left-moving'' and ``right-moving'' algebras. We will show that one can identify such structure from an approximate QCA as well, and use this to define a global QCA approximating $\alpha$. We describe this strategy in detail in \cref{sec:locally-rounding}.


\section{Approximate QCAs in one dimension: summary of results}\label{sec:summary of results}
We summarize our results on approximate QCAs in one dimension.
Here $\Omega$ is a finite interval or circle.
By an interval, we mean $\Omega = [a,b] := \{a,a+1,\ldots,b\}$ for $a,b \in \mathbb{Z}$, with the usual metric.
By a circle, we mean the quotient of an interval with endpoints identified.
Throughout this section, we set range~$r=1$, except when otherwise specified.
For instance, $x^+ := [x-1,x,x+1]$ when $x \in [a+1,b-1]$.
A larger range can be treated effectively by coarse-graining sites.

Let us begin with a statement for $\epsilon$-QCAs on the circle.
We show that these can always be approximated by some strict QCA, which is nearby in the local distance of \cref{eq:dist-local}:

\begin{theorem}\label{thm:circle-round}
Let $\alpha$ be an $\epsilon$-QCA on a finite circle~$\Omega$ with $|\Omega| \ge 8$.
Then there exists an exact QCA~$\tilde{\alpha}$ of range $3$ with $\distloc(\alpha,\tilde \alpha) = O(\eps).$
\end{theorem}
\noindent
We use $O(\eps)$ to indicate to an upper bound by~$c \epsilon$ for a universal constant~$c$.  We expect $|\Omega| \ge 8$ is an artifact of the proof and plan to address $|\Omega| \ge 4$ in the future.

We can also formulate a statement for $\epsilon$-locality-preserving homomorphisms on an interval.
These can be likewise approximated by strict locality-preserving homomorphisms:

\begin{theorem}\label{thm:round-homomorphism}
    Let $\Omega$ be an interval with $[a,b] \subset \Omega$ and $b-a \geq 8$.  Let
    $\alpha : \cA_{[a,b]} \to \cA_\Omega$ be $\epsilon$-locality-preserving and $\eps$-locally surjective homomorphism.
    Then there exists an exactly locality-preserving and exactly locally surjective homomorphism~$\tilde{\alpha} : \cA_{[a+2,b-3]} \to \cA_\Omega$ of range $2$, with $\distloc(\alpha|_{\cA_{[a+2,b-3]}},\tilde \alpha) = O(\eps).$
\end{theorem}
\noindent
These two results are proven at the end of Section~\ref{sec:locally-rounding}.

Let us recall the GNVW index.
While originally defined for exact QCAs on an infinite line~\cite{gross2012index}, the GNVW index can be defined for any homomorphism that is exactly locality-preserving and locally surjective, and moreover it is robust:
for two such $\alpha,\beta$, we have $\mathrm{Ind}(\alpha) = \mathrm{Ind}(\beta)$ whenever $\mathrm{dist}_\mathrm{loc}(\alpha,\beta) < \eps_0$, for some universal constant~$\eps_0$.
This is essentially the content of Remark~4.11 in Ref.~\cite{ranard2022converse}.
Moreover, the index obtained by first restricting to some sub-interval is independent of the choice of sub-interval.
We extend this index to the approximate case:

\begin{theorem}\label{thm:index}
There exists an assignment $\alpha \mapsto \mathrm{Ind}(\alpha)$ for $\epsilon$-locality-preserving and $\eps$-locally surjective homomorphisms~$\alpha$ as in \cref{thm:round-homomorphism}, with $\eps \leq \eps_0$ for a universal constant~$\eps_0$, such that the following holds:
(i) $\mathrm{Ind}(\alpha) = \mathrm{Ind}(\beta)$ when $\distloc(\alpha,\beta) \leq \eps_0$, and (ii) for $\eps=0$, $\mathrm{Ind}(\alpha)$ equals the GNVW index.
\end{theorem}

\noindent
This theorem is a straightforward corollary of \cref{thm:round-homomorphism}:
we simply define  $\mathrm{Ind}(\alpha):= \mathrm{Ind}_{\mathrm{GNVW}}(\tilde \alpha)$ by the GNVW index of the associated exact $\tilde \alpha$ supplied by \cref{thm:round-homomorphism}.
The result then follows directly from the robustness mentioned above.
Likewise, the index is invariant when restricting a homomorphism on a larger interval to a smaller interval.

We are not immediately supplied with an efficient algorithm for calculating $\mathrm{Ind}(\alpha)$ in the approximate case, because \cref{thm:round-homomorphism} relies on Ref.~\cite{kitaev2024almost},
whose proof is not constructive in the algorithmic sense.
However, an algorithm for constructing exact algebras from approximate algebras will be presented in future work~\cite{ranard26forthcoming}, with runtime polynomial in the dimension of the relevant algebras.
This will supply an algorithm for constructing the exact algebras in Ref.~\cite{kitaev2024almost}, and thereby for constructing the strictly locality-preserving homomorphism of \cref{thm:round-homomorphism}, and finally the index of \cref{thm:index}.

While \cref{thm:index} establishes an invariant, one might like to know it is a \textit{complete} invariant.
That is, we might want a classification of the form: two approximate QCAs can be connected by a path iff they have the same index.
With our weak notion of $\eps$-locality-preserving automorphisms, it may be difficult to formulate such a theorem using continuous paths.
However, it is straightforward to show the following, replacing continuous paths by finite sequences.
We formulate it on the circle, but the analogous statement holds for homomorphisms from an interval.
\begin{theorem}\label{thm:discrete path}
    There exist universal constants $\eps_0 > \eps_1 > 0$ such that the following holds for any two $\eps_1$-QCAs $\alpha, \beta : \cA_\Omega \to \cA_\Omega$ on a finite circle~$\Omega$ with at least $6$ sites:
    $\mathrm{Ind}(\alpha)=\mathrm{Ind}(\beta)$ iff there exists a finite sequence of $\eps_0$-QCAs $\alpha_0, \ldots, \alpha_n$ with $\alpha_0=\alpha$, $\alpha_n=\beta$, and $\mathrm{dist}_\mathrm{loc}(\alpha_i,\alpha_{i+1}) \leq \eps_0$ for $i \in [0,n-1]$.
\end{theorem}

\noindent
This follows from \cref{thm:circle-round}, \cref{thm:index}, and the fact that any two exact QCAs with the same GNVW index are connected by a continuous path of exact QCAs (Ref.~\cite{gross2012index}, Theorem 9).
Instead of taking $\eps_0$-QCAs, we can also take $\alpha_1,\dots,\alpha_{n-1}$ to be exact QCAs of range $2$.

\section{Notation and technical preliminaries}\label{sec:notation}
Throughout the text, all algebras are finite-dimensional unital $C^*$-algebras, and all subalgebras are unital.
In particular, every algebra is $*$-isomorphic to a direct sum of matrix algebras.

We first specify our notation for approximate equality and approximate inclusions.
Let $V$ be a normed vector space.
Let $x,y \in V$, and $\eps > 0$. We write
\begin{align*}
    x \approxEq{\eps} y \quad \Leftrightarrow \quad  \norm{x - y} \leq \eps.
\end{align*}
If $S,T \subset V$ are subspaces, we write
\begin{align*}
    S \approxSub{\eps} T \quad \Leftrightarrow \quad  \forall x \in S, \exists y \in T \text{ s.t. } \norm{x-y} \leq \eps \norm{x},
\end{align*}
and we write $S \approxEq\eps T$ if $S \approxSub\eps T$ and $T \approxSub\eps S$.
Similarly, we write $x \approxIn{\eps} S$ if there exists $y \in T$ such that $\norm{x-y} \leq \eps\norm{x}$, i.e. if $\mathrm{span}(x) \approxSub{\eps} S$.
Note that the condition for $x \approxEq{\eps} y$ does not scale the error by the norm of $x$ or $y$, in contrast to the approximate conditions involving linear subspaces.
The triangle inequality implies that in our notation, for $x,y,z \in V$ or $S,T,U \subset V$,
\begin{align*}
    x \approxEq{\delta} y \approxEq{\eps} z \quad &\Rightarrow \quad x \approxEq{\delta + \eps} z \\
    S \approxSub{\delta} T \approxSub{\eps} U  \quad &\Rightarrow \quad S \approxSub{\delta + \eps + \eps \delta} U.
\end{align*}
For the second line, note that for $x \in S$, there exists $y \in T$ such that $\norm{x-y} \leq \delta\norm{x}$, and hence $\norm{y} \leq (1 + \delta)\norm{x}$.
Various properties of approximate inclusions appear in \cref{app:near-inclusion-lemmas}.

Let $\cB \subset \cA$ be a subalgebra.
We let $\cB'$ denote the commutant of $\cB$ in $\cA$.
(We write $\cB' \cap \cA$ if we want to be explicit about $\cA$.)
Then $\cB'' = \cB$ by the double commutant theorem.
We now introduce some relevant maps on algebras.
Let
\begin{align}\label{eq:P-notation}
    P_\cB : \cA \to \cA, \qquad P_\cB(\cA) = \cB
\end{align}
denote the \emph{Hilbert-Schmidt projection} onto~$\cB$, that is, the orthogonal projection with respect to the Hilbert-Schmidt inner product, and let
\begin{align} \label{eq:T-notation}
 T_{\cB} : \cA \to \cA, \qquad T_{\cB}(x) = \int_{U(\cB)} u x u^* \, \mathrm{d}u
\end{align}
denote the ``twirling'' map, using the Haar probability measure on the group of unitaries in~$\cB$.
These are completely positive and unital (CPU).
We have
\begin{align}
     T_{\cB} = P_{\cB' \cap \cA}
\quad\text{and}\quad
     P_{\cB} = T_{\cB' \cap \cA},
 \end{align}
we use both notations for varying emphasis.
For factors $\cB, \cC$ and $\cA = \cB \ot \cC$,
 \begin{align*}
     P_\cB = d_\cC^{-1} \tr_\cC, \qquad  d_\cB=\sqrt{\dim(\cB)}.
 \end{align*}
In the context of QCAs, when working with an ambient tensor product algebra $\cA = \cA_{\Omega}$, with subalgebra $\cA_S \subset \cA$ for $S \subset \Omega$, we abbreviate $P_{\cA_S} = P_S$.

If $f : \cA \to \cB$ is an arbitrary linear map between algebras, we write $\norm f = \norm f_{\infty\to\infty}$, and denote its stabilization (the completely bounded norm) by
\begin{align*}
    \norm f_{\cb} := \sup_{\cC} \norm{f \ot \id_{\cC}},
\end{align*}
where $\id_{\cC}: \cC \to \cC$ is the identity map on $\cC$, and the supremum is over arbitrary algebras.
If $f$ is CPU, $\norm{f} = \norm{f}_{\cb} = 1$.

\section{Locally rounding QCAs}\label{sec:locally-rounding}
Given an approximate QCA, we refer to the process of finding a nearby exact QCA as ``rounding'' the QCA.
In this section, we study how to round an approximate 1D QCA in a local fashion.
That is, we modify the approximate QCA on a small interval to enforce exact locality on that region.
The input data will be a homomorphism mapping a small interval to a slightly larger interval that is approximately locality-preserving (in the sense of \cref{def:eps-local-homomorphism}).
For instance, it may be obtained from an approximate QCA on a larger one-dimensional system, with domain restricted to an interval.
We then use this local rounding result to establish the results announced in \cref{sec:summary of results}.
As explained earlier, we can set the range $r=1$ without loss of generality.

For our case of approximate QCAs in one dimension, our analysis takes advantage of the structure of exact 1D QCAs~\cite{gross2012index}. For an exact 1D QCA of range 1, one can define ``boundary algebras''~\cite{freedman2020classification},
\begin{align}\label{eq:strict-boundary-algebra}
    \tilde{\cR} = \alpha(\cA_{[1,2]}) \cap \cA_{[2,3]},
    \qquad
    \tilde{\cL} = \alpha(\cA_{[1,2]}) \cap \cA_{[0,1]},
\end{align}
with the property that
\begin{align*}
    \alpha(\cA_{[1,2]}) = \tilde{\cL} \otimes \tilde{\cR},
    \qquad
    \cA_{[1,2]} = \cL \otimes \cR,
\end{align*}
where $\cL = \alpha^{-1}(\tilde \cL)$ and $\cR = \alpha^{-1}(\tilde \cR)$.
See \cite[Thm.~4.1]{ranard2022converse} for one direct proof.
We also offer a different proof in \cref{rem:warmup} below, crafted as a warm-up to the approximate case.
Note that \cref{eq:strict-boundary-algebra} mirrors the notion of boundary algebra in \cite[Def.~3.5]{haah2023invertible}, though using a finite interval rather than a half-space.

To round a 1D \emph{approximate} QCA, the key will be to identify boundary algebras in the approximate setting.
\Cref{thm:local-rounding}, given an approximate QCA $\alpha$, will allow us find a nearby $\tilde \alpha$ that not only has strictly localized image $\alpha(\cA_{[1,2]}) \subset \cA_{[0,3]}$, but moreover satisfies the above properties for some choice of~$\cL, \cR$.
This factorization will allow us to locally round the approximate QCA on multiple intervals \emph{in parallel}, then glue these into a strict QCA on a larger region.
In the statement of the following theorem, we use the definitions of \cref{sec:QCA-definitions}, with range $r=1$ and working on the interval $\Omega=[-2,5]$ of $8$ sites.

\begin{theorem}\label{thm:local-rounding}
Let $\alpha : \cA_{[-1,4]} \to \cA_{[-2,5]}$ be a homomorphism satisfying
\begin{align*}
    \alpha(\cA_{[n,n+1]}) \subset_\eps \cA_{[n-1,n+2]} \quad\text{for } n\in \{-1,1,3\}, \qquad
    \cA_{[0,1]} \subset_\eps \alpha(\cA_{[-1,2]}), \qquad
    \cA_{[2,3]} \subset_\eps \alpha(\cA_{[1,4]}).
\end{align*}
Then there exist subalgebras
\begin{align*}
    \cL, \cR \subset \cA_{[1,2]}, \qquad [\cL,\cR]=0, \qquad \braket{\cL, \cR} = \cA_{[1,2]}
\end{align*}
such that $\alpha(\cL) \subset_{O(\eps)} \cA_{[0,1]}$ and $\alpha(\cR) \subset_{O(\eps)} \cA_{[2,3]}$.
As a consequence, there exists a homomorphism $\tilde{\alpha} : \cA_{[-1,4]} \to \cA_{[-2,5]}$ with
\begin{align*}
    \norm{\tilde\alpha - \alpha}_{cb} \leq O(\eps), \qquad
    \tilde \alpha(\cL) \subset \cA_{[0,1]}, \qquad
    \tilde \alpha(\cR) \subset \cA_{[2,3]}.
\end{align*}
\end{theorem}
\noindent
Again $O(\eps)$ refers to an upper bound by $c\epsilon$ for a universal constant $c$. We mention in \cref{cor:local-rounding} it is sufficient to assume $\alpha$ is an $\eps$-locality-preserving and $\eps$-locally-surjective homomorphism.

The central difficulty will be to make sense of the intersections of \cref{eq:strict-boundary-algebra} in the case of approximate QCAs.  Generically, if two subalgebras have some non-trivial intersection, but one of them is rotated slightly, the intersection becomes trivial. However, when two subalgebras $\cA, \cB$ have associated projections $P_\cA, P_\cB$ that approximately commute, we can form an exact algebra $\cC$ that acts as a robust version of the intersection $\cA \cap \cB$.
This is captured by \cref{lemma:intersections}.


\begin{remark}[Warm-up: exact case]\label{rem:warmup}
We first see how to prove \cref{thm:local-rounding} for the exact case~$\eps=0$.
This is essentially the ordinary structure theorem for 1D QCAs~\cite{gross2012index}, but the argument we give here will mirror the one used below for the approximate case.
We use the notation for projections~$P$ and twirling~$T$ from \cref{eq:P-notation,eq:T-notation}, with ambient algebra $\cA = \cA_{[-2,5]}$.
Motivated by \cref{eq:strict-boundary-algebra}, we define the following maps on~$\cA$:
\begin{align*}
    F_R := P_{\alpha(\cA_{[1,2]})} P_{[2,3]}, \qquad F_L := P_{\alpha(\cA_{[1,2]})}  P_{[0,1]}.
\end{align*}
We will first show that the projections making up $F_L,F_R$ commute, and hence $F_L, F_R$ are the orthogonal projections onto the algebras defined in \cref{eq:strict-boundary-algebra}.
To this end, we calculate
\begin{align}\label{eq:F_R-intro}
  F_R
= P_{\alpha(\cA_{[1,2]})} P_{[2,3]}
= P_{\alpha(\cA_{[1,2]})} P_{\alpha(\cA_{[1,4]})} P_{[2,3]}
= T_{\alpha(\cA_{[3,4]})} P_{\alpha(\cA_{[1,4]})} P_{[2,3]}
= T_{\alpha(\cA_{[3,4]})} P_{[2,3]},
\end{align}
where the last equality holds due to $\cA_{[2,3]} \subset \alpha(\cA_{[1,4]})$.
From the first equality we obtain $P_{[0,3]} F_R = F_R$, because $\alpha(\cA_{[1,2]}) \subset \cA_{[0,3]}$, and from the final expression we obtain $P_{[2,5]} F_R = F_R$, because $\alpha(\cA_{[3,4]}) \subset \cA_{[2,5]}$ and hence $P_{[2,5]} T_{\alpha(\cA_{[3,4]})} = T_{\alpha(\cA_{[3,4]})} P_{[2,5]}$.
Together, we see that $P_{[2,3]} F_R = F_R$.
That is:
\begin{align*}
  F_R
= P_{\alpha(\cA_{[1,2]})} P_{[2,3]}
= P_{[2,3]} P_{\alpha(\cA_{[1,2]})} P_{[2,3]}
= P_{[2,3]} P_{\alpha(\cA_{[1,2]})},
\end{align*}
where the last equality may be obtained from the first by taking adjoints.
Thus we see that $F_R$ is the product of \emph{commuting} projections onto $\cA_{[2,3]}$ and $\alpha(\cA_{[1,2]})$, hence it is in fact the projection on their intersection:
\begin{align*}
    F_R = P_{\tilde \cR}, \qquad \tilde \cR = \alpha(\cA_{[1,2]}) \cap \cA_{[2,3]} \subset \cA_{[2,3]}.
\end{align*}
Analogously, we see that
\begin{align*}
    F_L = P_{\tilde \cL}, \qquad \tilde \cL = \alpha(\cA_{[1,2]}) \cap \cA_{[0,1]} \subset \cA_{[0,1]},
\end{align*}
We have identified algebras $\tilde\cL \subset \cA_{[0,1]}$, $\tilde\cR \subset \cA_{[2,3]}$ such that $\tilde\cL \ot \tilde\cR \subset \alpha(\cA_{[1,2]})$, and will now show that the latter inclusion is in fact an equality.
To this end, from the final expression of \cref{eq:F_R-intro} and its left-hand analog, we see that
\begin{align*}
    P_{\tilde \cR}|_{\cA_{[2,3]}} = T_{\alpha(\cA_{[3,4]})}|_{\cA_{[2,3]}}, \qquad
    P_{\tilde \cL}|_{\cA_{[0,1]}} = T_{\alpha(\cA_{[-1,0]})}|_{\cA_{[0,1]}},
\end{align*}
Together we obtain
\begin{align*}
    P_{\tilde \cL \ot \tilde \cR}|_{\cA_{[0,3]}}
= T_{\alpha(\cA_{[-1,0]})} T_{\alpha(\cA_{[3,4]})}|_{\cA_{[0,3]}},
\end{align*}
and in particular
\begin{align*}
  P_{\tilde \cL \ot \tilde \cR} \alpha(\cA_{[1,2]})
= T_{\alpha(\cA_{[-1,0]})} T_{\alpha(\cA_{[3,4]})} \alpha(\cA_{[1,2]}) = \alpha(\cA_{[1,2]})
\end{align*}
because the twirls are over algebras in the commutant of $\alpha(\cA_{[1,2]})$.
Thus, $\alpha(\cA_{[1,2]}) \subset \tilde\cL \ot \tilde\cR$, hence
\begin{align*}
    \alpha(\cA_{[1,2]}) = \tilde \cL \otimes \tilde \cR.
\end{align*}
We can define $\cL = \alpha^{-1}(\tilde \cL)$ and $\cR = \alpha^{-1}(\tilde \cR)$ to conclude the proof.
\qed
\end{remark}

With this warm-up, we can prove \cref{thm:local-rounding}.
We repeatedly use the fact that for $\cB, \cC \subset \cA$,
\begin{align}\label{eq:PP-fact-proof}
 \cB \subset_\eps \cC  \implies P_{\cC} P_{\cB} =_{O(\eps)} P_{\cB}  =_{O(\eps)} P_{\cB} P_{\cC}
\end{align}
in the completely bounded norm, as in \cref{lemma:PP-subalgebra}.

\begin{proof}[Proof of \cref{thm:local-rounding}.]
We work within ambient algebra $\cA = \cA_{[-2,5]}$.
Throughout the proof, we will abbreviate $=_{O(\eps)}$ and $\subset_{O(\eps)}$ as $\approx$ and $\subsetsim$, respectively.
For maps between operator algebras these refer to completely bounded norm.

\emph{Step 1: Construction of boundary algebras.}
We will attempt to identify boundary algebras $\tilde{\cR}, \tilde{\cL}$ by using robust intersections in place of \cref{eq:strict-boundary-algebra}.
We focus first on $\tilde\cR$, and the construction of $\tilde\cL$ follows likewise.
To apply \cref{lemma:intersections}, we need
\begin{align}\label{eq:P23-commute}
    P_{\alpha(\cA_{[1,2]})}\,P_{[2,3]} \approx P_{[2,3]}\,P_{\alpha(\cA_{[1,2]})}.
\end{align}
We will achieve this by showing both sides are approximated by $P_{[2,3]}\,P_{\alpha(\cA_{[1,2]})} P_{[2,3]}$.
First we study
\begin{align}\label{eq:23-twirl}
    F_R
:= P_{\alpha(\cA_{[1,2]})} P_{[2,3]}
= P_{\alpha(\cA_{[1,2]})} P_{\alpha(\cA_{[1,4]})}  P_{[2,3]}
= T_{\alpha(\cA_{[3,4]})}  P_{\alpha(\cA_{[1,4]})} P_{[2,3]}
\approx T_{\alpha(\cA_{[3,4]})} P_{[2,3]},
\end{align}
where for the last step, we used the assumption $\cA_{[2,3]} \subsetsim \alpha(\cA_{[1,4]})$,
and hence $P_{\alpha(\cA_{[1,4]})} P_{[2,3]} \approx P_{[2,3]}$, using the first equality in \cref{eq:PP-fact-proof}.
From the first equality in \cref{eq:23-twirl} we obtain $P_{[0,3]}F_R \approx F_R$, because $\alpha(\cA_{[1,2]}) \subsetsim \cA_{[0,3]}$ and hence $P_{[0,3]} P_{\alpha(\cA_{[1,2]})} \approx P_{\alpha(\cA_{[1,2]})}$, again by the first equality in \cref{eq:PP-fact-proof}.
From the final expression in \cref{eq:23-twirl} we obtain $P_{[2,5]}F_R \approx F_R$, because $\alpha(\cA_{[3,4]}) \subsetsim \cA_{[2,5]}$, so we can use \cref{thm:near-inclusion-unitary} to identify an algebra $\cB \subset \cA_{[2,5]}$ with $\cB \approx \alpha(\cA_{[3,4]})$ and $T_{\alpha(\cA_{[3,4]})} \approx T_{\cB}$, hence $P_{[2,5]} T_{\cB} = T_{\cB} P_{[2,5]}$ and $P_{[2,5]} F_R \approx P_{[2,5]} T_{\cB} P_{[2,3]} = T_{\cB} P_{[2,3]} \approx F_R$.
Thus, we see that $P_{[2,3]} F_R = P_{[2,5]} P_{[0,3]} F_R \approx F_R$. That is:
\begin{align}\label{eq:F_R-1}
    F_R = P_{\alpha(\cA_{[1,2]})} P_{[2,3]} \approx P_{[2,3]} P_{\alpha(\cA_{[1,2]})} P_{[2,3]}.
\end{align}
We also want the same equation with the LHS reversed, so we can arrive at \cref{eq:P23-commute}.
We can argue similarly as above,%
\footnote{We cannot simply take the adjoint of \cref{eq:F_R-1} since this would result in a statement about the diamond norm, rather than the cb norm.}
but using the second equality in \cref{eq:PP-fact-proof} rather than the first, to obtain
\begin{align*}
    F_R'
:=  P_{[2,3]} P_{\alpha(\cA_{[1,2]})}
= P_{[2,3]} P_{\alpha(\cA_{[1,4]})} P_{\alpha(\cA_{[1,2]})}
=  P_{[2,3]}  P_{\alpha(\cA_{[1,4]})} T_{\alpha(\cA_{[3,4]})}
\approx  P_{[2,3]} T_{\alpha(\cA_{[3,4]})}
\end{align*}
and
\begin{align*}
F_R' = P_{[2,3]} P_{\alpha(\cA_{[1,2]})} \approx  P_{[2,3]} P_{\alpha(\cA_{[1,2]})} P_{[2,3]}.
\end{align*}
Combining this with \cref{eq:F_R-1} gives the approximate commutation in \cref{eq:P23-commute}.
We can thus apply \cref{lemma:intersections} to form a robust intersection of the subalgebras $\cA_{[2,3]} \subset \cA_{[-2,5]}$ and $\alpha(\cA_{[1,2]}) \subset \cA_{[-2,5]}$ to obtain an algebra $\hat \cR$ satisfying $\hat \cR \subsetsim \cA_{[2,3]}$,  $\hat \cR \subsetsim \alpha(\cA_{[1,2]})$, and
$P_{\hat \cR} \approx P_{\alpha(\cA_{[1,2]})} P_{[2,3]}$.
Using \cref{thm:near-inclusion-unitary} we can further modify $\hat \cR$ to yield an exact subalgebra $\tilde \cR \subset \cA_{[2,3]}$, while maintaining the above near-equalities to error $O(\eps)$.
That is, we have
\begin{align}\label{eq:tilde-R-inclusions}
    \tilde \cR \subset \cA_{[2,3]}, \qquad \tilde \cR \subsetsim \alpha(\cA_{[1,2]}),
    \qquad
    P_{\tilde \cR}
\approx P_{\alpha(\cA_{[1,2]})} P_{[2,3]}
\approx T_{\alpha(\cA_{[3,4]})} P_{[2,3]},
\end{align}
where the final expression follows from \cref{eq:23-twirl}.
By a symmetric argument, we can find
\begin{align*}
    \tilde \cL \subset \cA_{[0,1]}, \qquad \tilde \cL \subsetsim \alpha(\cA_{[1,2]}),
    \qquad
    P_{\tilde \cL}
\approx P_{\alpha(\cA_{[1,2]})} P_{[0,1]}
\approx T_{\alpha(\cA_{[-1,0]})} P_{[0,1]},
\end{align*}
Using \cref{lemma:simultaneous-near-inclusions}, we see that $\tilde \cL \ot \tilde \cR \subsetsim \alpha(\cA_{[1,2]})$.

\emph{Step 2: Approximate generation.}
We will now show that the converse inclusion $\alpha(\cA_{[1,2]}) \subsetsim \cL \ot \cR$ also holds, which will allow us to conclude that $\alpha(\cA_{[1,2]}) \approx \cL \ot \cR$.
Recall that above we used \cref{thm:near-inclusion-unitary} to obtain an algebra $\cB \subset \cA_{[2,5]}$ with $\cB \approx \alpha(\cA_{[3,4]})$ and $T_{\alpha(\cA_{[3,4]})} \approx T_{\cB}$.
Using this and the final expression in \cref{eq:tilde-R-inclusions}, we can write
\begin{align*}
    P_{\tilde \cR} & \approx  T_{\cB} P_{[2,3]}.
\end{align*}
All three are maps $\cA_{[-2,5]} \to \cA_{[-2,5]}$ that preserve the right half of the interval, $\cA_{[2,5]}$.
Thus we can restrict their domain and codomain to $\cA_{[2,5]}$ and obtain
 \begin{align*}
    P^{[2,5]}_{\tilde \cR} & \approx  T^{[2,5]}_{\cB}  P^{[2,5]}_{[2,3]},
\end{align*}
where the superscript $[2,5]$ indicates maps $\cA_{[2,5]} \to \cA_{[2,5]}$.
Because we are working with the completely bounded norm, we can tensor with the identity to obtain
\begin{align*}
    E_R
:= \id_{\cA_{[-2,1]}} \otimes P^{[2,5]}_{\tilde \cR}
\approx \id_{\cA_{[-2,1]}} \otimes  \left( T^{[2,5]}_{\cB}  P^{[2,5]}_{[2,3]} \right)
= T_{\cB} P_{[-2,3]}.
\end{align*}
Composing with $\alpha|_{\cA_{[1,2]}}$ on the right,
\begin{align}\label{eq:idPR}
    E_R \alpha|_{\cA_{[1,2]}}
    \approx T_{\cB} P_{[-2,3]} \alpha|_{\cA_{[1,2]}}
    \approx T_\cB \alpha|_{\cA_{[1,2]}}
    \approx T_{\alpha(\cA_{[3,4]})} \alpha|_{\cA_{[1,2]}}
    = \alpha|_{\cA_{[1,2]}}.
\end{align}
In the second step we used that $P_{[0,3]} P_{\alpha(\cA_{[1,2]})} \approx  P_{\alpha(\cA_{[1,2]})}$ and hence $P_{[0,3]} \alpha|_{\cA_{[1,2]}} \approx \alpha|_{\cA_{[1,2]}}$, due to $\alpha(\cA_{[1,2]}) \subsetsim \cA_{[0,3]}$ and \cref{eq:PP-fact-proof}.
By a symmetric argument one finds
 \begin{align} \label{eq:idPL}
   E_L \alpha|_{\cA_{[1,2]}} \approx \alpha|_{\cA_{[1,2]}}, \qquad
   E_L := P^{[-2,1]}_{\tilde \cL} \otimes \id_{\cA_{[2,5]}}.
\end{align}
Using \cref{eq:idPR,eq:idPL} we obtain
\begin{align*}
     E_L E_R \alpha|_{\cA_{[1,2]}} \approx \alpha|_{\cA_{[1,2]}}.
\end{align*}
Since $E_L E_R = P^{[-2,1]}_{\tilde \cL} \otimes P^{[2,5]}_{\tilde \cR} = P_{\tilde\cL \ot \tilde\cR}$ is a projection onto $\tilde \cL \otimes \tilde \cR$ and $\alpha$ is an isometry, it follows that
\begin{align*}
    \alpha(\cA_{[1,2]}) \subsetsim  \tilde \cL \otimes \tilde \cR.
\end{align*}
We already established the reverse near-inclusion, so we can conclude that
\begin{align*}
    \alpha(\cA_{[1,2]}) \approx \tilde{\cL} \otimes \tilde{\cR}.
\end{align*}
\emph{Step 3: Construction of $\cL,\cR,\tilde\alpha$.}
Using \cref{thm:near-inclusion-unitary} (twice), we see that there exists a unitary~$u \in \cA$ such that $\norm{u - I} \leq O(\eps)$ and
$\alpha(\cA_{[1,2]}) = u(\tilde{\cL} \otimes \tilde{\cR})u^*$.
As $\alpha$ is an isometric isomorphism onto its image, $\cL := \alpha^{-1}(u \tilde{\cL} u^*)$ and $\cR := \alpha^{-1}(u \tilde{\cR} u^*)$ are subalgebras with the desired properties:
they commute, generate $\cA_{[1,2]}$, and satisfy $\alpha(\cL) \approx \tilde{\cL} \subset \cA_{[0,1]}$ and $\alpha(\cR) \approx \tilde{\cR} \subset \cA_{[2,3]}$.
Furthermore, the homomorphism $\tilde{\alpha} : \cA_{[-1,4]} \to \cA_{[-2,5]}$ defined by $\tilde\alpha(x) := u^* \alpha(x) u$ satisfies $\tilde\alpha \approx \alpha$ as well as $\tilde\alpha(\cL) = \tilde\cL \subset \cA_{[0,1]}$ and $\tilde\alpha(\cR) = \tilde\cR \subset \cA_{[2,3]}$.
This concludes the proof.
\end{proof}

The following is a corollary to \cref{thm:local-rounding}.
\begin{corollary} \label{cor:local-rounding}
Let $\Omega$ be an interval or circle with $|\Omega| \geq 8$, with sites denoted $[-2,5] \subset \Omega$. Let $\alpha : \cA_{[-1,4]} \to \cA_\Omega$ be an $\eps$-locality-preserving and $\eps$-locally-surjective homomorphism.  Then the conclusions of \cref{thm:local-rounding} apply.
\end{corollary}
We continue using notation from the proof of \cref{thm:local-rounding}.
\begin{proof}
    By \cref{prop:extend-locality}, we have $\alpha(\cA_{[-1,4]}) \subsetsim \cA_{[-2,5]}$.
    We modify $\alpha$ to nearby $\tilde \alpha$ where this inclusion is exact.
    In particular, by  \cref{thm:near-inclusion-unitary}, there exists $v$ with $\norm{v-I} = O(\eps)$ such that $v \alpha(\cA_{[-1,4]}) v^* \subset \cA_{[-2,5]}$.
    Define $\tilde \alpha(x) = v \alpha(x) v^*$.  Then $\tilde \alpha(\cA_{[-1,4]}) \subset \cA_{[-2,5]}$, with~$\alpha \approx \tilde \alpha$ because~$\norm{v-I}=O(\eps)$.

    Then $\tilde \alpha$ is $O(\eps)$-locality-preserving and $O(\eps)$-locally-surjective.
    Because $\tilde \alpha(\cA_{[-1,4]}) \subset \cA_{[-2,5]}$, we may restrict the codomain to $\cA_{[-2,5]}$.
    Then the assumptions of \cref{thm:local-rounding} (involving various near-inclusions) apply to $\tilde \alpha$ by \cref{prop:extend-locality,prop:extend-loc-surjectivity}.
    Again using $\alpha \approx \tilde \alpha$, the conclusions of \cref{thm:local-rounding} applied $\tilde \alpha$ then imply the same conclusions applied to $\alpha$.
\end{proof}

The local rounding in \cref{thm:local-rounding} also works for larger regions, by implementing the local rounding separately on neighboring patches, then sewing together the results.

\begin{theorem*}[Restatement of \cref{thm:round-homomorphism}]
    Let $\alpha : \cA_{[a,b]} \to \cA_\Omega$ be a homomorphism, for $\Omega$ an interval, that is $\epsilon$-locality-preserving and $\eps$-locally surjective, with $[a,b] \subset \Omega$ and $b-a \geq 8$.
    Then there exists an exactly locality-preserving and exactly locally surjective homomorphism~$\tilde{\alpha} : \cA_{[a+2,b-3]} \to \cA_\Omega$ of range 2 with $\distloc(\alpha|_{\cA_{[a+2,b-3]}},\tilde \alpha) = O(\eps).$
\end{theorem*}
The intervals here are not optimally chosen.  If $[a+2,b-2]$ has even cardinality, then in fact the argument below produces a homomorphism $\hat{\alpha} : \cA_{[a+2,b-2]} \to \cA_\Omega$ satisfying the same properties.
\begin{proof}
It is enough to prove the following stronger statement: if $[a+2,b-2]$ has even cardinality, then there exists a homomorphism
\begin{align*}
    \hat{\alpha} : \cA_{[a+2,b-2]} \to \cA_\Omega
\end{align*}
which is exactly locality-preserving and locally surjective, and satisfies
\begin{align*}
    \distloc(\alpha|_{\cA_{[a+2,b-2]}},\hat{\alpha}) = O(\eps).
\end{align*}
Indeed, this immediately implies the stated theorem:
if $[a+2,b-2]$ has odd cardinality, then applying the stronger statement to the interval $[a,b-1]$ yields such a homomorphism on $\cA_{[a+2,b-3]} \subset \cA_{[a+3,b-3]}$.
Thus assume from now on that $[a+2,b-2]$ has even cardinality.  Write
\begin{align*}
    [a+2,b-2] = \bigcup_{i=1}^m B_i, \qquad B_i := [a+2i,a+2i+1], \qquad m = \frac{b-a-3}{2}
\end{align*}
as a disjoint union of adjacent pairs, and let
\begin{align*}
    C_i := [a+2i-1,a+2i], \qquad i=1,\ldots,m+1,
\end{align*}
denote the shifted adjacent pairs, where each $B_i$ sits between $C_i$ and $C_{i+1}$.

Fix $i \in \{1,\ldots,m\}$.  Apply \cref{thm:local-rounding} (via \cref{cor:local-rounding}) to the restriction of $\alpha$ to the six-site interval $[a+2i-2,a+2i+3]$.
We obtain commuting subalgebras
\begin{align*}
    \cL_i,\cR_i \subset \cA_{B_i}, \qquad [\cL_i,\cR_i] = 0, \qquad \braket{\cL_i,\cR_i}=\cA_{B_i},
\end{align*}
and a homomorphism $\beta_i : \cA_{B_i} \to \cA_\Omega$ such that
\begin{align*}
    \beta_i =_{O(\eps)} \alpha|_{\cA_{B_i}}, \qquad
    \beta_i(\cL_i) \subset \cA_{C_i}, \qquad
    \beta_i(\cR_i) \subset \cA_{C_{i+1}}
\end{align*}
in the completely bounded norm.
Denote
\begin{align*}
    \tilde \cL_i := \beta_i(\cL_i) \subset \cA_{C_i}, \qquad
    \tilde \cR_i := \beta_i(\cR_i) \subset \cA_{C_{i+1}}.
\end{align*}

We now make neighboring image algebras commute exactly.  For each $i=1,\ldots,m-1$, note that $\cR_i$ and $\cL_i$ commute, so we have the exact commutation
\begin{align*}
    \alpha(\cR_i) \subset \alpha(\cL_i)'.
\end{align*}
Next, $\tilde \cR_i, \tilde \cL_{i+1} \subset \cA_{C_{i+1}}$, and since $\tilde \cR_i =_{O(\eps)} \alpha(\cR_i)$ and $\tilde \cL_{i+1} =_{O(\eps)} \alpha(\cL_{i+1})$ and $\alpha(\cR_{i})$ and $\alpha(\cL_{i+1})$ commute, we have that $\norm{[x,y]} \leq O(\eps \norm{x}\norm{y})$ for all $x \in \tilde\cR_i$ and $y \in \tilde\cL_{i+1}$. 
By \cref{thm:near-inclusion-unitary}, this implies
\begin{align*}
    \tilde \cR_i \subset_{O(\eps)} \tilde \cL_{i+1}' \cap \cA_{C_{i+1}}.
\end{align*}
Hence again by \cref{thm:near-inclusion-unitary}, there exists a unitary $u_i \in \cA_{C_{i+1}}$ with $\norm{u_i-I}=O(\eps)$ such that
\begin{align*}
    \hat \cR_i := u_i \tilde \cR_i u_i^* \subset \tilde \cL_{i+1}' \cap \cA_{C_{i+1}}.
\end{align*}
Now replace $\beta_i$ by 
\begin{align*}
    \hat \beta_i := u_i \beta_i(\cdot) u_i^*.
\end{align*}
Since $u_i \in \cA_{C_{i+1}}$, this does not change the exact inclusion $\beta_i(\cL_i)\subset \cA_{C_i}$, while it sends~$\beta_i(\cR_i)$ to~$\hat \cR_i \subset \cA_{C_{i+1}}$.  Moreover,
\begin{align}\label{eq:hat beta close to alpha}
    \hat \beta_i =_{O(\eps)} \beta_i =_{O(\eps)} \alpha|_{\cA_{B_i}}.
\end{align}

The homomorphisms $\hat \beta_i$ may now be glued together.  Indeed, if $|i-j|>1$, then~$\hat \beta_i(\cA_{B_i})$ and~$\hat \beta_j(\cA_{B_j})$ have disjoint support, hence commute.  If $j=i+1$, then $\hat \beta_i(\cR_i)=\hat \cR_i$ commutes with $\hat \beta_{i+1}(\cL_{i+1})=\tilde \cL_{i+1}$ by construction, while the remaining pieces lie in the disjoint pairs $C_i$ and $C_{i+2}$.  Thus the algebras~$\hat \beta_i(\cA_{B_i})$ commute pairwise, and since
\begin{align*}
    \cA_{[a+2,b-2]} = \bigotimes_{i=1}^m \cA_{B_i},
\end{align*}
the maps $\hat \beta_i$ combine to a homomorphism
\begin{align*}
    \hat \alpha : \cA_{[a+2,b-2]} \to \cA_\Omega
\end{align*}
defined by its action on generators, $\hat \alpha(\cA_{B_i}) = \beta_i(\cA_{B_i})$.

Recall that by the definition of the local distance (\cref{eq:dist-local}), to bound $\distloc(\alpha|_{\cA_{[a+2,b-2]}},\hat \alpha)$, it suffices to bound the distance on single-site algebras.
Each site of $[a+2,b-2]$ lies in a pair $B_i$, and on each such pair we have $\hat \alpha =_{O(\eps)} \alpha$ in local distance by \cref{eq:hat beta close to alpha}. Therefore we obtain
\begin{align*}
    \distloc(\alpha|_{\cA_{[a+2,b-2]}},\hat \alpha)=O(\eps).
\end{align*}
Since $\alpha$ is $\eps$-locally surjective, this also implies that $\hat \alpha$ is $O(\eps)$-locally surjective.
In particular, we have
\begin{align*}
    \cA_{C_{i}} \subset_{O(\eps)} \hat \alpha(\cA_{B_{i-1} \cup B_i}),
\end{align*}
and since $\hat{\alpha}(\cL_{i-1}) \subseteq \cA_{C_{i-1}}$ and $\hat{\alpha}(\cR_{i}) \subseteq \cA_{C_{i+1}}$ this implies
\begin{align*}
    \cA_{C_{i}} \subset_{O(\eps)} \braket{\hat\alpha(\cR_{i-1}), \hat\alpha(\cL_i)}.
\end{align*}
Conversely, by construction $\braket{\hat\alpha(\cR_{i-1}), \hat\alpha(\cL_i)} \subset \cA_{C_i}$, and we conclude from  (for instance) \cref{thm:near-inclusion-unitary} that we must have 
\begin{align}\label{eq:hat alpha locally surjective}
    \braket{\hat\alpha(\cR_{i-1}), \hat\alpha(\cL_i)} = \cA_{C_i}.
\end{align}

When the pairs $B_i$ are considered as single sites $i$ with unit spacing,
then $\hat \alpha$ is exactly locality-preserving and locally surjective with range 1, with respect to this coarse-graining.
Indeed, the locality-preserving property holds due to
\begin{align*}
    \hat \alpha(\cA_{B_i}) = \hat \beta_i(\cA_{B_i}) \subset \cA_{C_i \cup C_{i+1}} \subset \cA_{B_{i-1} \cup B_i \cup B_{i+1}},
\end{align*}
and the local surjectivity holds due to \cref{eq:hat alpha locally surjective}.
When individual sites are considered, without coarse-graining into disjoint pairs~$B_i$, then we conclude the weaker desired statement: $\hat \alpha$ is locality-preserving and locally surjective with range 2.
\end{proof}

Finally we turn to circles.

\begin{theorem*}[Restatement of \cref{thm:circle-round}]
Let $\alpha$ be an $\epsilon$-QCA on a finite circle~$\Omega$ with $|\Omega| \ge 8$.
Then there exists an exact QCA~$\tilde{\alpha}$ of range $3$ with $\distloc(\alpha,\tilde \alpha) = O(\eps).$
\end{theorem*}
\begin{proof}
We will prove that when $|\Omega|$ is even, there exists an exact QCA~$\tilde{\alpha}$ of range $2$ with $\distloc(\alpha,\tilde \alpha) = O(\eps)$
The case for $|\Omega|$ odd can be reduced to the even case, by fixing two adjacent sites and coarse-graining them into a single site.

For even $|\Omega|$, partition $\Omega$ into pairs of adjacent sites.
Choose a pair of sites and label them $[1,2]$.
Because $|\Omega| \ge 8$, they have neighborhoods $[1,2] \subset [-1,5] \subset [-2,5] \subset \Omega$ such that $\alpha|_{[-1,5]}$ is $\eps$-locality-preserving and $\eps$-locally-surjective.
Then we can apply \cref{cor:local-rounding} to the homomorphism $\alpha|_{\cA_{[-1,5]}}$ to obtain subalgebras $\cL, \cR \subset \cA_{[1,2]}$.
We repeat this separately for each disjoint pair of sites partitioning $\Omega$.
Then the proof proceeds nearly identically to that of \cref{thm:round-homomorphism}. Note that the objects that appear in the proof (such as the algebras $\cL_i$, $\cR_i$ and the maps $\hat \beta_i$) are all constructed in a local way.  The circle geometry does not matter, because the modified algebras  $\hat \cL_i$, $\hat \cR_i$ are constructed via independent modifications in each pair $i$.
Similarly, the argument for the properties of the glued map $\hat \alpha$ is purely local.
\end{proof}

\section{Approximate Algebras and Intersections}\label{sec:approx algebras}

\subsection{Approximate \texorpdfstring{$C^*$}{C*}-algebras}
Kitaev~\cite{kitaev2024almost} introduces an approximate version of a unital $C^*$-algebra, where the multiplication is only associative up to a small error.
We first review this notion.

\begin{definition}\label{dfn:approx algebra}
    An $\eps$-$C^*$-algebra for $0 \leq \eps < 1$ is a Banach space $\cA$, with a bilinear multiplication map $x , y \mapsto xy$, a conjugate linear involution $x \mapsto x^*$ and an approximate unit element $e \in \cA$. The multiplication satisfies
    \begin{align*}
        \norm{xy}            \leq (1 + \eps) \norm{x}\norm{y} \qquad\text{ and }\qquad
        \norm{(xy)z - x(yz)} \leq \eps \norm{x} \norm{y} \norm{z},
    \end{align*}
    for all $x, y, z \in \cA$.
    The involution satisfies
    \begin{align*}
        \norm{x^*} = \norm{x} \qquad \text{ and } \qquad \norm{x^* x} \geq (1 - \eps) \norm{x}^2
    \end{align*}
    for all $x \in \cA$, and $(xy)^* = y^* x^*$ for all $x,y \in \cA$.
    Finally, the unit satisfies
    \begin{align*}
        \norm{xe - x} \leq \eps\norm{x}, \qquad \norm{ex - x} \leq \eps\norm{x}, \qquad \abs{\norm{e} - 1} \leq \eps
    \end{align*}
    for all $x \in \cA$, and $e^* = e$.
    We say that $\cA$ is \emph{exactly unital} if $xe = x = ex$ and $\norm{e} = 1$.
\end{definition}

We will mostly consider \emph{finite dimensional} $\eps$-$C^*$-algebras.
Note that an $\eps$-$C^*$-algebra for $\eps = 0$ is the same as a $C^*$-algebra, which we will occasionally call an \emph{exact} $C^*$-algebra to avoid confusion.

We can also define approximate homomorphisms between approximate or exact algebras.

\begin{definition}\label{dfn:approx homomorphism}
    Let $\cA_i$ be an $\eps_i$-$C^*$-algebra for $i = 1,2$ and let $0 \leq \delta  < 1$. A \emph{$\delta$-homomorphism} is a bounded linear map
    \begin{align*}
        \phi : \cA_1 \to \cA_2
    \end{align*}
    such that $\phi(x^*) = \phi(x)^*$ for all $x \in \cA_1$ and
    \begin{align*}
        \norm{\phi(xy) - \phi(x)\phi(y)} \leq \delta \norm{x}\norm{y} \qquad \text{ and } \qquad \norm{\phi(e) - e} \leq \delta
    \end{align*}
    for all $x,y \in A_1$.
    A \emph{$\delta$-isomorphism} is a bijective $\delta$-homomorphism.
\end{definition}
Note that in this definition, the $\eps_i$ can equal zero, so we can have an approximate homomorphism between exact algebras.
In this case, it is well-known that for many $C^*$-algebras, approximate homomorphisms are close to exact ones~\cite{johnson1988approximately}, see
Lemma 8.2 of~Ref.~\cite{kitaev2024almost} for a proof in the finite-dimensional setting.

\begin{theorem}\label{thm:make homomorphism exact}
    If $\phi : \cA \to \cB$ is a $\delta$-homomorphism and $\cA$ and $\cB$ are exact $C^*$-algebras, there exists an exact homomorphism $\psi : \cA \to \cB$ with $\norm{\psi - \phi} = \bigO(\delta)$.
\end{theorem}

Here we collect a few basic properties of approximate algebras, see \cite{kitaev2024almost} for proofs.

\begin{lemma}
    Let $\cA$ be a $\eps$-$C^*$-algebra.
    \begin{enumerate}
        \item We have $\norm{x^*x} \leq (1 + \eps)\norm{x}^2$ for all $x \in \cA$.
              \item\label{it:nearby exact unit} There exists a $\bigO(\eps)$-isomorphism of $\cA$ with an exactly unital $\bigO(\eps)$-$C^*$-algebra.
        \item If $\phi : \cA \to \cB$ is a $\delta$-isomorphism, then $\phi^{-1}$ is a $\bigO(\delta)$-isomorphism.
    \end{enumerate}
\end{lemma}

In light of \cref{it:nearby exact unit} we will from now on only consider exactly unital $\eps$-$C^*$-algebras.
The main result of Ref.~\cite{kitaev2024almost} is that every $\eps$-$C^*$-algebra is close to an exact $C^*$-algebra:

\begin{theorem}[Theorem 2.3 in Ref.~\cite{kitaev2024almost}] \label{thm:algebras-kitaev}
    There exists $\eps_0 >0$ such that for any finite dimensional $\eps$-$C^*$-algebra $\cA$ with $\eps < \eps_0$, there exists a $C^*$-algebra $\cB$ and a $\bigO(\eps)$-isomorphism
    \begin{align*}
        \phi : \cB \to \cA.
    \end{align*}
\end{theorem}
\noindent
In particular, the map $\phi$ is a linear bijection.

We sketch some ideas in the proof for interested readers.
It loosely follows the classification of finite dimensional $C^*$-algebras as matrix algebras.
Suppose $\cA$ is an exact finite dimensional $C^*$-algebra.
One can way to show that $A$ is isomorphic to a matrix algebra, is by starting from a maximal commutative subalgebra, which is spanned by a collection of projections $p_i$.
One can then show that $p_i \cA p_j$ is a subspace of dimension either 0 or 1. In the second case, let $P_{ij}$ be a basis vector for $p_i \cA p_j$.
One can then show that the $P_{ij}$ form a collection of matrix units, explicitly reconstructing $\cA$ as a matrix algebra.
For $\eps$-$C^*$-algebras, Ref.~\cite{kitaev2024almost} uses a similar strategy, with additional subtleties.
In particular, in order to construct the approximate version of the maximal commutative subalgebra, one needs to find nontrivial approximate projections.
In the exact case, one can do so through spectral calculus, but this is not quite available in the approximate case.
This is solved in Ref.~\cite{kitaev2024almost} through a topological argument, constructing approximate projections as the fixed point of a certain continuous map.
This allows one in principle to iteratively construct a collection of approximate projections $p_i$, and subspaces $p_i \cA p_j$ as in the exact case.
However, doing so naively accumulates error, leading to an error bound which is dimension dependent.
However, this can be resolved using an error-reduction technique: given a $\delta$-homomorphism from an exact algebra $\cB$ to an $\eps$-$C^*$-algebra $\cA$, there also exists a $\bigO(\eps)$-homomorphism from $\cB$ to $\cA$.

We gave an abstract definition of an $\eps$-$C^*$-algebra.
One concrete way in which such algebras may appear is as subspaces $\cA \subset B(\cH)$ which are approximately closed under multiplication.
\begin{definition}\label{def:eps-subalgebra}
    A subspace $S \subset B(\cH)$ is \emph{$\eps$-closed under multiplication} if for any $x, y \in S$ there exists $z \in S$ such that $\norm{xy-z} \leq \eps \norm{x}\norm{y}.$
    Let $Q : B(\cH) \to B(\cH)$ be a linear map.
    We call $(S, Q)$ an \emph{$\eps$-subalgebra} if moreover $Q^2 = Q$, $Q(x^*) = Q(x)^*$ for all $x \in B(\cH)$ and
    \begin{align*}
        \mathrm{Image}(Q) = S \qquad \text{ and } \qquad \norm{Q} \leq 1 + \eps.
    \end{align*}
\end{definition}

In this case, the subspace $S$ is contained in an ambient space $B(\cH)$ with an exact multiplication. Using only the property that $xy$ is approximately in  $S$, without using $Q$,
one might try to construct an $\eps$-$C^*$-algebra structure on $S$ by defining a new bilinear product $x, y \mapsto x \cdot y \in S$ such that $\norm{x \cdot y - xy} \leq \eps \norm{x}\norm{y}$.
However, it is not obvious how to do so for an arbitrary subspace which is approximately closed under multiplication.
However, if  $S$ is an $\eps$-subalgebra, defined as above and thus equipped with $Q$, then $S$ can be given the structure of an $\eps$-$C^*$-algebra by defining a product on $S$ through $x \cdot y = Q(xy)$.

\begin{lemma}\label{lemma:eps-subalgebra}
    Suppose $(S \subset B(\cH), Q)$ is an $\eps$-subalgebra (\cref{def:eps-subalgebra}).
    Then the following holds:
    \begin{enumerate}
        \item\label{it:approx product} The product $x \cdot y = Q(xy)$ gives $S$ the structure of an $\eps$-$C^*$-approximate algebra.
        \item\label{it:close subalgebra} If $S$ is finite dimensional, there exists an exact subalgebra $\cA \subset B(\cH)$ such that $\cA \approxEq{\bigO(\eps)} S$.
    \end{enumerate}
\end{lemma}

\begin{proof}
    Let $x, y \in S$, and let $z \in S$ such that $\norm{xy - z} \leq \eps \norm{x}\norm{y}$.
    Then
    \begin{align*}
        \norm{xy - x \cdot y} & = \norm{xy - Q(xy)} \leq \norm{xy - Q(z)} + \norm{Q(z) - Q(xy)}                 \\
                              & \leq \norm{xy - z} + \norm{Q}\norm{xy - z} \leq (2 + \eps)\eps \norm{x}\norm{y}
    \end{align*}
    using $Q(z) = z$ as $z \in S$.
    We now check the conditions in \cref{dfn:approx algebra} in order to prove \ref{it:approx product}.
    Firstly, for any $x,y \in S$,
    \begin{align*}
        \norm{x \cdot y} = \norm{Q(xy)} \leq \norm{Q}\norm{xy} \leq (1 + \eps)\norm{x}\norm{y}.
    \end{align*}
    Next, from the fact that $\norm{x \cdot y - xy} = \bigO(\eps)\norm{x}\norm{y}$ it follows that for $x, y, z \in S$
    \begin{align*}
        \norm{(x \cdot y)\cdot z - x \cdot (y \cdot z)} & = \bigO(\eps)\norm{x}\norm{y}\norm{z}.
    \end{align*}
    Since $Q(x^*) = Q(x)^*$ for all $x \in B(\cH)$, we have $(x \cdot y)^* = y^* \cdot x^*$.
    We have
    \begin{align*}
        \norm{x^* \cdot x} \geq \norm{x^*x} - \norm{x^* x - x^* \cdot x} = (1 - \bigO(\eps))\norm{x}^2.
    \end{align*}
    Finally, it is clear that $I$ is an exact unit for $S$.
    Next, we show \ref{it:close subalgebra}.
    By \cref{thm:algebras-kitaev} there exists a $\bigO(\eps)$-isomorphism
    $\phi : \cB \to S$ for a $C^*$-algebra $\cB$ (with respect to the product structure of \ref{it:approx product}).
    Let $\psi$ be the composition of $\phi$ with the inclusion of $S$ into $B(\cH)$.
    Since $\norm{x \cdot y - xy} = \bigO(\eps)$, it is immediate that $\psi$ is also a $\bigO(\eps)$-homomorphism.
    By \cref{thm:make homomorphism exact} there exists an exact homomorphism $\theta : \cB \to \cB(\cH)$ with
    \begin{align}\label{eq:theta close to psi}
        \norm{\theta - \psi} = \bigO(\eps).
    \end{align}
    We let $\cA = \mathrm{Image}(\theta)$.
    Note that since $\theta$ is an isometry, by \cref{eq:theta close to psi} $\psi$ is approximately an isometry in the induced operator norm.
    This, together with \cref{eq:theta close to psi} implies that $\cA \approxEq{\bigO(\eps)} S$.  
\end{proof}

\subsection{Approximately idempotent maps}

\begin{lemma}[Rounding nearly idempotent maps]\label{lemma:round}
    Suppose that $F : L(\cH) \to L(\cH)$ linear satisfies $\norm{F - F^2} \leq \eps \leq \frac18$ in some submultiplicative norm~$\norm{\cdot}$.
    Then there exists $\round(F) := Q$ such $Q^2 = Q$, and
    \begin{align*}
        \norm{F - Q} \leq O(\eps) (\norm{F} + 1)
    \end{align*}
    Moreover, if $F(x^*)=F(x)^* \; \forall x \in  L(\cH)$, then $Q$ can be chosen with the same property.
\end{lemma}
\begin{proof}
    Define $X := 2F-1$.
    Note that the assumption implies that
    \begin{align*}
        \norm{X^2 - I}
        = \norm{4F^2 - 2F - 2F + I - I}
        = 4 \norm{F^2 - F}
        \leq 4\eps \leq \frac12
    \end{align*}
    and hence (now we use submultiplicativity of the norm)
    \begin{align*}
        (X^2)^{-1/2}
    \end{align*}
    can be defined via its power series, which is absolutely convergent around~1; from this we also see that
    \begin{align*}
        \norm{(X^2)^{-1/2} - I } = O(\eps).
    \end{align*}
    Then we can define
    \begin{align*}
        \mathrm{sign}(X) := X (X^2)^{-1/2}
    \end{align*}
    and get
    \begin{align*}
        \norm{X - \mathrm{sign}(X)} \leq \norm X \norm{I - (X^2)^{-1/2}} = O(\eps) \norm X
    \end{align*}
    From the power series it is also clear that $\mathrm{sign}(X)^2 = I$.  Also, in the case $F(x^*)=F(x)^*$, then  $\mathrm{sign}(X)(x^*)=\mathrm{sign}(X)(x)^*$.
    Finally, we define
    \begin{align*}
        Q := \frac12 (I + \mathrm{sign}(X))
    \end{align*}
    Then we have $Q^2 = Q$ and
    \begin{align*}
        \norm{F - Q}
        = \norm{\frac12 (I + X) - \frac12 (I + \mathrm{sign}(X))}
        = \frac12 \norm{X - \mathrm{sign}(X)}
        = O(\eps)\norm{X}
        = O(\eps) (\norm{F} + 1).
    \end{align*}
\end{proof}

\subsection{Approximate intersections}

When two subalgebras $\cA, \cB$ have associated Hilbert-Schmidt projections $P_\cA, P_\cB$ that approximately commute, we can form an exact algebra $\cC$ that acts as a robust version of the intersection $\cA \cap \cB$.
This is captured by \cref{lemma:intersections} below. Note the exact intersection $\cA \cap \cB$ can vary wildly under perturbations of $\cA, \cB$; indeed, it is empty under generic small rotations of either algebra.
In contrast, the ``intersection'' $\cC$ of  \cref{lemma:intersections} is robust to perturbations of $\cA$ and $\cB$, as described below.

One open question not addressed by the lemma is whether one can find nearby algebras $\tilde{\cA}, \tilde{\cB}$ such that their conditional expectations exactly commute.

For finite-dimensional $C^*$-algebra $\cM$ with unital subalgebra $\cA \subset \cM$ we use $P_\cA : \cM \to \cM$ to denote the conditional expectation given by Hilbert-Schmidt projection onto $\cA$. The theorem and proof below may be interpreted with either (1) the assumptions and conclusions both using the operator norm for the linear maps, or (2) the assumptions and conclusions both using the $cb$-norm for the linear maps.

\begin{theorem}[Approximate intersections from approximately commuting expectations]\label{lemma:robust-intersection} \label{lemma:intersections}
Let $\cM$ be a finite-dimensional $C^*$-algebra and let $P_{\cA},P_{\cB}:\cM\to\cM$ be the Hilbert-Schmidt projections onto unital subalgebras $\cA, \cB \subset \cM$, so that
$\cA=\image(P_{\cA})$ and $\cB=\image(P_{\cB})$. Assume
\[
\|P_{\cA}P_{\cB}-P_{\cB}P_{\cA}\|\le \eps.
\]
Then there exist a unital subalgebra $\cC\subset \cM$ such that
\begin{enumerate}[label=(\alph*)]
\item $\cC\subset_{O(\eps)}\cA$ and $\cC\subset_{O(\eps)}\cB$.
\item For all $\delta\ge 0$ and all $x\in\cM$,
\[
x\in_\delta \cA \ \text{ and }\ x\in_\delta \cB \ \Longrightarrow\ x\in_{O(\delta)+O(\eps)} \cC.
\]
\item  $\norm{P_{\cC}- P_{\cA}P_{\cB}} = O(\eps)$
for Hilbert-Schmidt projection $P_\cC$.
\end{enumerate}
\end{theorem}

\noindent
We refer to the algebra $\cC$ produced above as an ``approximate intersection'' of $\cA$ and $\cB$.
It is not uniquely specified by the conditions $(a,b,c)$.
However, it is approximately unique in the sense that if $\cC$ and $\tilde \cC$ satisfy the conditions $(a,b)$ above, then $\cC =_{O(\epsilon)} \tilde \cC$.
Moreover, $\cC$ is robust to perturbations of the inputs: if $(\cA_1, \cB_1)$ and $( \cA_2,  \cB_2)$ separately satisfy the assumptions of \cref{lemma:robust-intersection}, with $\cA_1=_\eps \cA_2$ and $\cB_1 =_\eps \cB_2$, then the outputs $\cC_1, \cC_2$ of \cref{lemma:robust-intersection} satisfy $\cC_1 =_{O(\eps)}  \cC_2$.
\begin{proof}
To summarize the argument, $P_\cA P_\cB$ is approximately idempotent, so the image of $\mathrm{round}(P_\cA P_\cB)$ is approximately closed under multiplication, allowing us to find a nearby exact algebra $\cC$ by \cref{lemma:eps-subalgebra}.

Denote
\begin{align}
    F:=P_{\cA}P_{\cB}.
\end{align}
Since $P_{\cA}^2=P_{\cA}$ and $P_{\cB}^2=P_{\cB}$,
\[\
\norm{F^2-F} = \norm{P_{\cA}(P_{\cB}P_{\cA}-P_{\cA}P_{\cB})P_{\cB}} \le \eps.
\]
Also $F$ is CPU and $*$-preserving, hence $\|F\|=1$.

Assume $\eps\le 1/8$.  (Otherwise, the conclusions involving $O(\eps)$ bounds become vacuous.) Apply \cref{lemma:round} to obtain an idempotent $*$-map
$Q=\round(F)$ with $Q^2=Q$, $Q(x^*)=Q(x)^*$, and
\[
\|Q-F\|\le O(\eps),\qquad \|Q\|\le 1+O(\eps).
\]
In particular,
\begin{align} \label{eq:FQetc}
    Q  =_{O(\eps)} F =  P_\cA P_\cB =_{O(\eps)} P_\cB P_\cA.
\end{align}
Let $S:=\image(Q)$. From the above, it follows $S \subset_{O(\eps)} \cA, \cB$.

Next, $S$ is $O(\eps)$-closed under multiplication. Let $s,t\in S$. From $s\in_{O(\eps)}\cA$ and $t\in_{O(\eps)}\cA$ we get $st\in_{O(\eps)}\cA$, and similarly $st\in_{O(\eps)}\cB$.  From \cref{eq:FQetc}, it follows
\[
\|st-Q(st)\|\le O(\eps)\,\|s\|\,\|t\|.
\]

Thus $(S,Q)$ is an $O(\eps)$-subalgebra in the sense of \cref{def:eps-subalgebra}. By \cref{lemma:eps-subalgebra}, there exists an exact unital subalgebra $\cC\subset \cM$ with
$\cC\approxEq{\bigO(\eps)}S$.  Then from $S \subset_{O(\eps)} \cA, \cB$ we have conclusion (a),
\begin{equation}
    \cC\subset_{O(\eps)}\cA, \qquad \cC\subset_{O(\eps)}\cB.
\end{equation}

To show (b), take $\norm{x} = 1$ without loss of generality and assume $x\in_\delta \cA$ and $x\in_\delta \cB$. (Recall that for elements of normed vector spaces we denote $\norm{a-b} \le \eps$ as $a=_\eps b$.)
Then $x =_{O(\delta)} P_\cA(x)$ and  $x =_{O(\delta)} P_\cB(x)$, so  $x =_{O(\delta)} Q(x) \in S =_{O(\eps)} \cC$. It follows $x \in_{O(\eps) + O(\delta)} \cC.$

To show (c), note $\cC \subset_{O(\eps)} \cA$ implies $P_\cC P_\cA =_{O(\eps)} P_\cC$, and likewise $P_\cC P_\cB =_{O(\eps)} P_\cC$, by \cref{lemma:PP-subalgebra}.  Also note $P_\cC Q =_{O(\eps)} Q$, using $S =_{O(\eps)}  \cC$.  Then
\begin{align}
    P_\cC =_{O(\eps)} P_\cC P_\cA P_\cB =_{O(\eps)} P_\cC Q =_{O(\eps)} Q =_{O(\eps)} F,
\end{align}
as desired.
\end{proof}

\appendix
\crefalias{section}{appendix}

\section{Near-inclusion lemmas}\label{app:near-inclusion-lemmas}
We record several facts about near inclusions of algebras. These are essentially finite-dimensional specializations of results in Ref.~\cite{ranard2022converse}.
Recall that for normed vector spaces $\mathcal{A}, \mathcal{B}$, we write $\mathcal{A} \subset_\epsilon \mathcal{B}$ to mean that for every $a \in \mathcal{A}$ there exists $b \in \mathcal{B}$ such that $\|a-b\| \le \epsilon \|a\|.$
We also write $a \in_\epsilon \mathcal{B}$ when there exists $b \in \mathcal{B}$ such that $\|a-b\| \le \epsilon \|a\|.$  For finite-dimensional $C^*$-subalgebras $\cA, \cB$, we use the notation $C^*(\cA,\cB)$ to denote the algebra they generate.  For linear maps on normed vector spaces, we use operator norm.

\begin{lemma}[Near inclusions and commutators (Ref.~\cite{ranard2022converse}, Lemma 2.4)] \label{lemma:near-inclusion-commutators}
Let $\mathcal{A},\mathcal{B} \subset \mathcal{M}$ be unital subalgebras of a finite-dimensional $C^*$-algebra $\mathcal{M}$.
If $\mathcal{A} \subset_\epsilon \mathcal{B}'$, then for all $a \in \mathcal{A}$ and $b \in \mathcal{B}$,
\[
\|[a,b]\| \le 2\epsilon \|a\|\,\|b\|.
\]
Conversely, if for all $a \in \mathcal{A}$ and $b \in \mathcal{B}$ one has
\[
\|[a,b]\| \le \epsilon \|a\|\,\|b\|,
\]
then $\mathcal{A} \subset_\epsilon \mathcal{B}'$.
\end{lemma}

\begin{lemma}[Near inclusion of commutants (Ref.~\cite{ranard2022converse}, Lemma 2.5)]\label{lemma:near-inclusion-commutants}
Let $\cA, \cB \subset \cM$ be unital subalgebras of a finite-dimensional $C^*$-algebra $\cM$. If $\cA \subset_\epsilon \cB$, then
\begin{align}
\cB' \subset_{2\epsilon} \cA'.
\end{align}
\end{lemma}

The below also appears in Ref.~\cite{ranard2022converse}, Theorem 2.6, where the proof exposition is slightly corrected.
\begin{theorem}[Near inclusions to exact inclusions (Ref.~\cite{christensen1980near}, Theorem 4.1)]\label{thm:near-inclusion-unitary}
Let $\mathcal{A},\mathcal{B} \subset \mathcal{M}$ be unital subalgebras of a finite-dimensional $C^*$-algebra $\mathcal{M}$, and suppose $\mathcal{A} \subset_\epsilon \mathcal{B}$ with $\epsilon \le 1/64$. Then there exists a unitary $u \in C^*(\mathcal{A},\mathcal{B})$ such that
\[
u^* \mathcal{A} u \subset \mathcal{B},
\qquad
\|u-I\| \le 12\epsilon.
\]
\end{theorem}

\begin{lemma}[Projections for near-inclusions] \label{lemma:PP-subalgebra}
There exists a universal constant $c$ such that the following holds. Let $\cA, \cB \subset \cM$ be unital subalgebras of a finite-dimensional $C^*$-algebra $\cM$.  If $\cB \subset_\eps \cA$, then
     \begin{align}
      \norm{P_{\cA} P_{\cB} - P_{\cB} }_{cb} \le c\eps , \qquad \norm{P_{\cB} P_{\cA} - P_{\cB}}_{cb}   \le c\eps.
 \end{align}
 in the completely bounded operator norm.
\end{lemma}
\begin{proof}
    While there is a more elementary proof, this follows directly from application of \cref{thm:near-inclusion-unitary} to  $\cB \subset_\eps \cA$
\end{proof}

\begin{lemma}[Local error to global error for homomorphisms (Ref.~\cite{ranard2022converse}, Lemma 2.7)]\label{lemma:local-to-global-hom}
Let $\alpha_1,\alpha_2 : \mathcal{A} \to \mathcal{B}$ be injective unital $*$-homomorphisms between finite-dimensional unital $C^*$-algebras. Suppose $\mathcal{A}_1,\ldots,\mathcal{A}_n \subset \mathcal{A}$ are pairwise commuting unital subalgebras generating $\mathcal{A}$, and define
\[
\epsilon = \sum_{i=1}^n \|(\alpha_1-\alpha_2)|_{\mathcal{A}_i}\|.
\]
If $\epsilon<1$, then
\[
\|\alpha_1-\alpha_2\|
\le
2\sqrt{2}\,\epsilon \left(1+\sqrt{1-\epsilon^2}\right)^{-1/2}
\le
2\sqrt{2}\,\epsilon.
\]
\end{lemma}

\begin{lemma}[Simultaneous near inclusions (Ref.~\cite{ranard2022converse}, Lemma B.3]\label{lemma:simultaneous-near-inclusions}
Let $\mathcal{A}_1,\ldots,\mathcal{A}_n,\mathcal{B} \subset \mathcal{M}$ be unital subalgebras of a finite-dimensional $C^*$-algebra $\mathcal{M}$. Suppose the algebras $\mathcal{A}_i$ pairwise commute and satisfy $\mathcal{A}_i \subset_{\epsilon_i} \mathcal{B}$ for each $i$. Writing $\epsilon = \sum_i \epsilon_i$, one has
\[
\mathcal{B}' \subset_{2\epsilon} C^*(\mathcal{A}_1,\ldots,\mathcal{A}_n)', \qquad C^*(\mathcal{A}_1,\ldots,\mathcal{A}_n) \subset_{4\epsilon} \mathcal{B}.
\]
\end{lemma}


\begin{proposition}[Almost-unitaries are close to unitaries]\label{prop:almost-unitary-polar}
Let $\cB$ be a unital finite-dimensional $C^*$-algebra, and let $x \in \cB$ satisfy
\[
\|x^*x-I\| \le \epsilon < \frac12.
\]
Then $x^*x$ is invertible, so the polar part
\[
u := x(x^*x)^{-1/2}
\]
is a unitary in $\cB$, and
\begin{align}
\|x-u\| \le 2\epsilon.
\end{align}
\end{proposition}

\begin{proof}
Since $\|x^*x-I\|<1$, the spectrum of $x^*x$ is contained in $[1-\epsilon,1+\epsilon] \subset (0,\infty)$, so $x^*x$ is invertible. Hence $u=x(x^*x)^{-1/2} \in \cB$ is well-defined and satisfies
\[
u^*u=(x^*x)^{-1/2}x^*x(x^*x)^{-1/2}=I.
\]
Since $x^*x$ is invertible, also $xx^*$ is invertible, so $u$ is unitary. Note
\[
\|x-u\|
=
\|x-x(x^*x)^{-1/2}\|
\le
\|x\|\,\|I-(x^*x)^{-1/2}\|.
\]
Also,
\[
\|x\|^2=\|x^*x\| \le 1+\epsilon,
\]
so $\|x\| \le (1+\epsilon)^{1/2}$. Finally, since $\sigma(x^*x)\subset [1-\epsilon,1+\epsilon]$,
\[
\|(x^*x)^{-1/2}-I\|
\le
\max_{t\in[1-\epsilon,1+\epsilon]} |t^{-1/2}-1|
=
(1-\epsilon)^{-1/2}-1.
\]
Therefore
\[
\|x-u\|
\le
(1+\epsilon)^{1/2}\left((1-\epsilon)^{-1/2}-1\right)
=
\frac{(1+\epsilon)^{1/2}-(1-\epsilon)^{1/2}}{(1-\epsilon)^{1/2}}.
\]
With $\epsilon<1/2$,  we obtain the desired $\|x-u\|\le 2\epsilon$.
\end{proof}

\begin{proposition}[A unitary near a subalgebra is near a unitary in the subalgebra]
\label{prop:unitary-near-subalgebra}
Let $\cB \subset \cA$ be unital finite-dimensional $C^*$-algebras, and let $u \in \cA$ be unitary. If $u \in_\delta \cB$ and $(2+\delta)\delta < 1/2$, then there exists a unitary $v \in \cB$ such that
\[
\|u-v\| \le (5+2\delta)\delta.
\]
In particular, if $\delta \le 1/8$, then
\[
\|u-v\| \le 6\delta.
\]
\end{proposition}

\begin{proof}
Choose $x \in \cB$ with $\|u-x\|\le \delta$. Since $u$ is unitary,
\[
x^*x-I = x^*x-u^*u = x^*(x-u)+(x^*-u^*)u.
\]
Therefore,
\[
\|x^*x-I\|
\le
\|x\|\,\|x-u\|+\|x^*-u^*\|\,\|u\|
\le
\|x\|\delta+\delta.
\]
Also $\|x\| \le \|u\|+\|x-u\| \le 1+\delta$, so
\[
\|x^*x-I\| \le (2+\delta)\delta.
\]
By assumption, $(2+\delta)\delta<1/2$, so \cref{prop:almost-unitary-polar}
applies to $x$, giving a unitary
\[
v := x(x^*x)^{-1/2} \in \cB
\]
such that
\[
\|x-v\| \le 2\|(x^*x)-I\| \le 2(2+\delta)\delta.
\]
Hence
\[
\|u-v\| \le \|u-x\|+\|x-v\|
\le
\delta + 2(2+\delta)\delta
=
(5+2\delta)\delta.
\]
If $\delta \le 1/8$, we obtain the desired $\|u-v\| \le 6\delta.$
\end{proof}

\bibliographystyle{plain}
\bibliography{references}
\end{document}